\documentclass[12pt,peerreview,onecolumn]{IEEEtranTCOM}
\pdfoutput =1
\usepackage{amsthm,amsmath,amssymb,bbm}
\usepackage[utf8]{inputenc} 
\usepackage{tabularx}
\usepackage{hhline}
\usepackage{tikz}
\usepackage{mathtools}
\usepackage{cuted}
\usepackage{graphicx}
\usepackage{listings,lstautogobble}
\usepackage{xcolor}

\ifCLASSOPTIONpeerreview
\renewcommand{\baselinestretch}{2}
\fi

\definecolor{RoyalBlue}{cmyk}{1, 0.50, 0, 0}

\newtheorem{remark}{Remark}
\newtheorem{theorem}{Theorem}

\newtheorem{prop}{Proposition}
\newtheorem{definition}{Definition}

\newcommand{\AuthorOne}{Ross~Pure}
\newcommand{\AuthorTwo}{Salman~Durrani}
\newcommand{\AuthorThree}{Fei Tong}
\newcommand{\AuthorFour}{Jianping Pan}
\newcommand{\ThankOne}{R. Pure and S. Durrani are with the Research School of Electrical, Energy and Materials Engineering, The Australian National University, Canberra, ACT 2615, Australia. F. Tong is with the Department of Control Science and Engineering, Zhejiang University, Hangzhou, China. J. Pan is with the Department of Computer Science, University of Victoria, BC, Canada. Corresponding author email: salman.durrani@anu.edu.au}

\begin{document}
\title{Distance Distribution Between Two Random Nodes in Arbitrary Polygons}
\author{\authorblockN{\AuthorOne,~\AuthorTwo,~\AuthorThree~and \AuthorFour\thanks{\ThankOne}}}
\maketitle
%
\begin{abstract} 
Distance distributions are a key building block in stochastic geometry modelling of wireless networks and in many other fields in mathematics and science. In this paper, we propose a novel framework for analytically computing the closed form probability density function (PDF) of the distance between two random nodes each uniformly randomly distributed in respective arbitrary (convex or concave) polygon regions (which may be disjoint or overlap or coincide). The proposed framework is based on measure theory and uses polar decomposition for simplifying and calculating the integrals to obtain closed form results. We validate our proposed framework by comparison with simulations and published closed form results in the literature for simple cases. We illustrate the versatility and advantage of the proposed framework by deriving closed form results for a case not yet reported in the literature. Finally, we also develop a Mathematica implementation of the proposed framework which allows a user to define any two arbitrary polygons and conveniently determine the distance distribution numerically.
\end{abstract}


\ifCLASSOPTIONpeerreview
    \newpage
\fi

\section{Introduction}

\subsection{Background and Motivation}
Heterogeneous cellular networks are a key building block of present fourth and future fifth generation cellular networks~\cite{Andrews2014}. A key characteristic of heterogeneous cellular networks is the irregular and dense deployment of macro and small cell base stations. Due to this, the coverage areas of base stations (i.e., the cell boundaries) form a Voronoi tessellation. Each cell in the Voronoi tessellation is an arbitrarily shaped polygon. This can be confirmed by examining actual base station deployment data which has been reported in recent papers. See for instance, Fig. 2 in \cite{Andrews2016} and Fig. 2 in \cite{Andrews2011}.

In the last decade, stochastic geometry has emerged as a powerful analytically tractable technique to accurately model heterogeneous cellular networks~\cite{Haenggi2013,ElSawy2017}. Stochastic geometry is an abstraction based modelling technique, i.e., instead of using actual base station and user locations, it uses random locations of base stations and users. The stochastic geometry framework is built upon two main building blocks~\cite{Haenggi2013}: (i) the moment generating function of the aggregate interference and (ii) the distance distributions, i.e., the probability density function (PDF) or equivalently the cumulative distribution function (CDF), of the nodes (nodes can be base stations and/or users). Both of these building blocks are dependant on the locations of the nodes, which are seen as realizations of some spatial point process~\cite{Andrews2013}. Typically, node locations are assumed to follow infinite homogeneous Poisson point process (PPP), i.e., the network is assumed to be infinitely large and have an infinite number of nodes. In reality, the number of nodes is fixed and finite and the network region is finite as well~\cite{Guo2014}. In this paper, we focus on the distance distributions for arbitrarily shaped polygon regions, which model typical cells in heterogeneous cellular networks.\footnote{Note that although the focus of this paper is on distance distributions used in stochastic geometry, such distance distributions are also used in many other fields, such as mathematics, physics, forestry, operations research and material sciences~\cite{mathai_1999,Moltchanov-2012,Coon-2012,Li2016}.}
		
\subsection{Related Work}
In general, there are two types of distance distributions that are needed in the stochastic geometry modelling~\cite{tong_2017,khalid_2013}: (i) the distribution of the distance between a given reference node (located inside or outside the cell) and a random node located inside a cell, and (ii) the distribution of the distance between two random nodes (located in the same or different cells). An example of the former is the nearest neighbour distance distribution when the reference node (e.g., a base station) is located at the center of the cell. An example of the latter is the distribution of the distance between two randomly located device or machine type nodes in the same or different cells.

In the last decade, many works have investigated the first type of distance distributions, i.e., there is one fixed reference point and one uniformly randomly distributed point in some region. In this case, the most common strategy for obtaining the distance distributions is by first computing the CDF and then differentiating to find the PDF. Computing the CDF amounts to finding the area of intersection of a given circle and the region in question, so computing this area of intersection is the main mathematical challenge. For example, if the fixed point is inside the region, then closed forms have been obtained for when the region is simple, such as a square~\cite{pirinen_2006}, and for more general cases such as regular polygons~\cite{khalid_2013} and arbitrary convex polygons~\cite{baltzis_2013}. The most general result in the literature is~\cite{ahmadi_2014}, which can compute distance distributions in the case where the fixed point is located anywhere (outside or inside the region) and the region is an arbitrary polygon (convex or concave). This method was modified and implemented in Mathematica in~\cite{pure_2015} to obtain the closed form expressions given an input fixed point and arbitrary polygon region. Note that, in general, we can approximate any region with arbitrary precision using a polygon with a sufficient number of sides, so this last result can still be useful for cases where the region is not a polygon, e.g., a weighted Voronoi cell for a heterogeneous cellular network with base stations having different transmitting powers where the cell boundaries have curved lines (arcs).

The focus of this paper is the second type of distance distributions, i.e., where there are two random points each uniformly randomly distributed in respective regions (which may be disjoint or overlap or coincide). This case is significantly more complicated and consequently the literature is not as complete. Obtaining closed form expressions for the distance distributions in this case has almost exclusively been limited to when the two regions coincide (i.e., two uniformly randomly distributed points in the same region), and only in relatively simple cases. For example, results in the literature include circles~\cite{mathai_1999,sinanovic_2008}, triangles~\cite{basel_2012}, squares~\cite{khalid_2014}, rectangles~\cite{zhuang_2010} and regular polygons~\cite{basel_2014}. One case where closed forms have been found for shapes that do not coincide can be found in~\cite{zhuang_2010} where the case of two neighbouring rectangles, and also in a limited scope two diagonal rectangles, is calculated. Recently, a result that applies to the most general case of arbitrary regions (not necessarily coinciding) was obtained in~\cite{tong_2017,tong_2017_2}.
		
\begin{table*}[t]
			\caption{Comparison of existing techniques for computing distance distributions between two random points.} \label{table:comparison}
				\begin{tabularx}{\textwidth}{|c|p{30mm}|X|X|p{30mm}|}
					\hline
					Technique & Main Idea & Advantage & Limitation & Closed Form \\ \hhline{|=|=|=|=|=|}
					\cite{mathai_1999} & First principle derivations &  & Limited applications and extensions to new cases & Coincident circles, coincident rectangles \\ \hline
             		\cite{zhuang_2010} & Direct manipulation of the underlying random variables &  & Limited applications and extensions to new cases & Rectangles \\ \hline
             		~\cite{basel_2014} & Chord length distance distributions & & & Convex polygons \\ \hline
					\cite{tong_2017} & Kinematic measure & Numerical solutions applicable to many cases & Closed forms in special cases only & Coincident circles, coincident triangles \\ \hline
					This paper & Measure theory and polar decomposition & Numerical solutions applicable to all cases and closed forms achievable for arbitrary polygons &  & Arbitrary polygons \\ \hline
			\end{tabularx}
\end{table*}
		
A tabular summary and comparison of the main results is provided in Table~\ref{table:comparison}. It can be seen that the approaches vary and the results are generally not derived from a common mathematical framework.
The most comprehensive framework is provided in~\cite{tong_2017}, which uses a technique called kinematic measure. The approach in~\cite{tong_2017} is applicable to any arbitrary regions, including regions with holes. However, in most general cases, the closed form results are not possible and the result is only known in integral form and needs to be implemented numerically. These limitations of the known results and techniques motivates our work in this paper.

		
\subsection{Contributions}
The main contributions of this work are:
			\begin{itemize}
				\item We present a general framework for analytically computing the closed form PDF of the distance between two random nodes each uniformly randomly distributed in respective arbitrary (convex or concave) polygon regions (which may be disjoint or overlap or coincide). The proposed framework is based on measure theory\footnote{In mathematics, a measure is a generalization of the concepts of length, area, and volume~\cite{stein_vol3}.} and uses polar decomposition for simplifying and calculating the integrals to obtain closed form results.

				\item We provide examples to show how the proposed framework is able to find closed form results for simple cases reported in the literature and a new case not reported in closed form in the literature, i.e., two disjoint triangles.

				\item We develop a Mathematica implementation which implements the proposed framework and is able to numerically calculate the distance distribution for arbitrary (convex or concave) polygon regions. Simulation results verify the accuracy of the derived distance distribution results.
			\end{itemize}
	
\subsection{Notations}
The following is the notation that will be used throughout the paper. Arbitrary measure spaces will be denoted using the triples $(X, \mathcal{M}, \mu)$ and $(Y, \mathcal{M}, \nu)$, and arbitrary subsets of these spaces will be denoted by $A$, $B$, and $E$. $\lambda_n$ denotes the Lebesgue measure on $\mathbb{R}^n$. The function $\mathbbm{1}_E$ will denote the characteristic function corresponding to the set $E$. Arbitrary probability measures will be denoted by $\mathbb{P}$. Regions in the plane $\mathbb{R}^2$ will be denoted by calligraphic letters; $\mathcal{A}$ and $\mathcal{B}$ denote arbitrary regions, $\mathcal{P}$ and $\mathcal{Q}$ denote polygons, and $\mathcal{T}$ denotes triangles.
		
\subsection{Paper organisation}
This paper is organized as follows. Section~\ref{sec:maths} summarises the measure theory concepts used in this work. Section~\ref{sec:proposed} presents the proposed mathematical framework. Section~\ref{sec:results} presents examples that illustrate the application of the proposed framework and also discusses the Mathematica implementation. Finally, Section~\ref{sec:conc} concludes the paper.

\section{Mathematical Framework}\label{sec:maths}
	In this section we outline a rigorous development of probability theory which forms the basis of the proposed formulation of distance distributions. The probability theory builds on measure and integration theory, which we summarize first.
	
	\subsection{Measure Spaces}
		A measure space with a corresponding set $X$ is such that we can assign a ``measure'' or ``size'' to certain subsets $A \subset X$. That is, we define a function  $\mu$ called the \textit{measure} such that $\mu(A)$ is the measure of $A$. The definition is as follows~\cite{stein_vol3}.
		
		\begin{definition}
			A measure space, denoted by the triple $(X, \mathcal{M}, \mu)$, is a set $X$ that has two associated objects:
			\begin{enumerate}
				\item A $\sigma$-algebra $\mathcal{M}$ of sets that are considered to be ``measurable''.
				\item A function $\mu : \mathcal{M} \to [0, \infty]$, with the property that if $E_1, E_2, \ldots \subset \mathcal{M}$ are disjoint, then
				\begin{align}
					\mu\left(\bigcup_{n=1}^\infty E_n\right) = \sum_{n=1}^\infty \mu(E_n). \label{eq:meas_decomp}
				\end{align}
			\end{enumerate}
		\end{definition}

		A relevant example is the measure space $(\mathbb{R}^n, \mathcal{B}_n, \lambda_n)$, which is our familiar setting of volume in $\mathbb{R}^n$; when $n = 1$, the measure gives the length of a set, when $n = 2$ it gives the area, when $n = 3$ it gives the volume, and so on. The $\sigma$-algebra $\mathcal{B}_n$ is called the Borel $\sigma$-algebra, the details of which are not important here because in our case we will not encounter the case of subsets of $\mathbb{R}^n$ not included in $\mathcal{B}_n$. That is to say, any sets we are going to consider in the context of distance distributions will be measurable and so we need not be concerned about problems of non-measurability.
	
		If we have two measure spaces $(X, \mathcal{M}_1, \mu)$ and $(Y, \mathcal{M}_2, \nu)$ it is also possible to define a product measure space $(X \times Y, \mathcal{M}, \mu \times \nu)$, where $\mathcal{M}$ is generated by $\mathcal{M}_1 \times \mathcal{M}_2$. This is done by first specifying that if $A \in \mathcal{M}_1$ and $B \in \mathcal{M}_2$ then $(\mu \times \nu)(A \times B) = \mu(A)\nu(B)$. Given this it is possible to extend $\mu \times \nu$ to all of $\mathcal{M}$ uniquely\footnote{We only have uniqueness when the two measure spaces are $\sigma$-finite, which is true for all the cases we are interested in here.}. For example, the product measure space of $(\mathbb{R}^n, \mathcal{B}_n, \lambda_n)$ and $(\mathbb{R}^m, \mathcal{B}_m, \lambda_m)$ is $(\mathbb{R}^n \times \mathbb{R}^m, \mathcal{B}', \lambda_n \times \lambda_m)$. We can identify $\mathbb{R}^n \times \mathbb{R}^m$ with $\mathbb{R}^{n+m}$, and it turns out that similarly $\lambda_n \times \lambda_m = \lambda_{n+m}$.
	
		We can also define integration with respect to a measure. The fundamental property of the integral used in its construction is that for a set $E \in \mathcal{M}$ we have
		\begin{align}
			\int \mathbbm{1}_E \mathrm{d}\mu = \int_E \mathrm{d}\mu = \mu(E), \label{eq:meas_int}
		\end{align}
		where $\mathbbm{1}_E$ is the characteristic function of the set $E$, which is 1 on $E$ and 0 elsewhere.
		
	\subsection{Fubini's Theorem}
		A useful theorem is Fubini's theorem, which generalises the notion of iterating integrals. The statement of the theorem is as follows~\cite{stein_vol3}.

		\begin{theorem}
			(Fubini) Let $(X, \mathcal{M}_1, \mu)$ and $(Y, \mathcal{M}_2, \nu)$ be measure spaces, and let $f : X \times Y \to \mathbb{R}$ be defined by $(x, y) \mapsto f(x, y)$. If $f$ is integrable (i.e., the integral with respect to $\mu \times \nu$ is finite), and if we have the slice function
			\begin{align}
				f^y : X &\to \mathbb{R}\\
				x &\mapsto f(x, y)
			\end{align}
			then
			\begin{align}
				\int_{X \times Y} f \mathrm{d}(\mu \times \nu) = \int_Y \left(\int_X f^y \mathrm{d}\mu\right)\mathrm{d}\nu
			\end{align}
		\end{theorem}

		This theorem generalises the common practice of splitting up an integral in $\mathbb{R}^2$ or $\mathbb{R}^3$ into the separate coordinates, and integrating over each coordinate iteratively.
		
	\subsection{Probability Spaces}
		Using the definition of a measure space, we can define a probability space $(\Omega, \Sigma, \mathbb{P})$, which is simply a measure space for which the measure of the entire set is unity. This is summarised in the following definition~\cite{stein_vol4}.
		\begin{definition}
			A probability space is a measure space $(\Omega, \Sigma, \mathbb{P})$ such that $\mathbb{P}(\Omega) = 1$.
		\end{definition}
		The three objects of a probability space can be thought of in the following way.
		\begin{itemize}
			\item $\Omega$ is a set, representing possible events.
			\item $\Sigma$ is a $\sigma$-algebra of $\Omega$, which can be thought of as the set of subsets of $\Omega$ for which we can define a probability of occurrence.
			\item $\mathbb{P} : \Sigma \to [0, 1]$ is a probability measure, where $\mathbb{P}(A)$ is the probability that any element contained in the event $A$ occurs.
		\end{itemize}

		For example, if we have a region $\mathcal{A} \subset \mathbb{R}^2$, we can define a uniformly randomly chosen point from $\mathcal{A}$ using the probability space $(\mathbb{R}^2, \mathcal{B}, \mu_{\mathcal{A}})$, where we define
		\begin{align}
			\mu_{\mathcal{A}}(E) = \frac{\lambda_2(E \cap \mathcal{A})}{\lambda_2(\mathcal{A})}. \label{eq:p_meas}
		\end{align}
		We can also express probabilities in terms of integrals; using \eqref{eq:meas_int} we can write
		\begin{align}
			\mathbb{P}(E) = \int \mathbbm{1}_E \mathrm{d}\mathbb{P}.
		\end{align}
		Furthermore, if we want the probability of an event $A$ conditioned on an event $B$, we write this conditional probability as $\mathbb{P}(A \mid B)$, where
		\begin{align}
			\mathbb{P}(A \mid B) = \frac{\mathbb{P}(A \cap B)}{\mathbb{P}(B)}. \label{eq:cond}
		\end{align}
		Additionally, we can define probabilities using a PDF. If $X = \mathbb{R}^n$ for some $n \in \mathbb{N}$ and we are given a PDF $f : \mathbb{R}^n \to \mathbb{R}$ we know that the probability of $E$ occurring is
		\begin{align}
			\mathbb{P}(E) = \int_{E} f\mathrm{d}\lambda_n.
		\end{align}
		In measure theoretic terms, we say that $f$ is the Radon-Nikodym derivative corresponding to the measures $\mathbb{P}$ and $\lambda_n$~\cite{royden_1968}. This is useful because it allows us to compute probabilities using integrals from standard Riemann integration theory, instead of integrating with respect to an abstract measure. For example, the PDF corresponding to the example probability measure \eqref{eq:p_meas} is $f(x) = \frac{1}{\lambda_2(\mathcal{A})}$.

		
	\subsection{Shoelace Formula}
		Finally, a useful theorem that will be used in some computations for later examples is the Shoelace Formula~\cite{weisstein_shoelace}, which is a convenient method to compute the area of a polygon given its vertices. It is defined as follows~\cite{weisstein_shoelace,pure_2015}.
		\begin{theorem} \label{thm:shoelace}
			(Shoelace Formula) Let $\mathcal{P} \subset \mathbb{R}^2$ be a non self-intersecting polygon with vertices $(x_1, y_1), \ldots,\allowbreak (x_n, y_n)$. Then the area of $\mathcal{P}$ is
			\begin{align}
				\lambda_2(\mathcal{P}) = \frac{1}{2}\sum_{i=1}^n x_i y_{i+1} - x_{i+1} y_i,
			\end{align}
			where we understand $x_{n+1}$ and $y_{n+1}$ to be $x_1$ and $y_1$ respectively.
		\end{theorem}
		
\section{Mathematical Formulation of Distance Distributions}~\label{sec:proposed}

	\subsection{Distance Distributions}
		We will now present the measure theoretic formulation of distance distributions. Given two regions $\mathcal{A} ,\mathcal{B} \subset \mathbb{R}^2$ and two respective probability measures $\mu_{\mathcal{A}}$ and $\mu_{\mathcal{B}}$, we independently choose two points, one from each region. For these two points, the distance between them is a random variable, and hence admits a CDF $F(r)$, which evaluated at $r > 0$ is by definition the probability that the two points are within a distance $r$ of each other.

		To express this CDF in the language of measure theory, we let $C_r = \{(a, b) \in \mathcal{A} \times \mathcal{B} : |b - a| < r\}$, i.e., the set of all pairs of points from $\mathcal{A}$ and $\mathcal{B}$ that are within a distance of $r$ of each other. The CDF is then given by the product measure
		\begin{align}
			F(r) = (\mu_{\mathcal{A}} \times \mu_{\mathcal{B}})(C_r).
		\end{align}
		Written in integral form this becomes
		\begin{align}
			F(r) = \int\mathbbm{1}_{C_r}\mathrm{d}(\mu_{\mathcal{A}} \times \mu_{\mathcal{B}}) = \int_{C_r}\mathrm{d}(\mu_{\mathcal{A}} \times \mu_{\mathcal{B}}).
		\end{align}
		If we know the PDFs for each of the points, say $f_{\mathcal{A}}$ and $f_{\mathcal{B}}$ respectively, using the independence of the points we can instead write
		\begin{align}
			F(r) = \int \mathbbm{1}_{C_r}f_{\mathcal{A}} f_{\mathcal{B}} \mathrm{d}(\lambda_2 \times \lambda_2).\label{eq:fund}
		\end{align}
		
		\eqref{eq:fund} is the starting point for computing distance distributions in our framework because we usually explicitly know the PDFs $f_{\mathcal{A}}$ and $f_{\mathcal{B}}$, and \eqref{eq:fund} allows us to work with standard Riemann integrals.
		
	\subsection{Proposed Idea}
		We will present a new framework for computing the PDF of the random distance between two uniformly distributed points in two regions. The result is summarised in the following theorem.
	
		\begin{theorem} \label{thm:new}
		Let $\mathcal{A}, \mathcal{B} \subset \mathbb{R}^2$ be two regions and $x \in \mathcal{A}$ and $y \in \mathcal{B}$ be two uniformly distributed random points. Then the PDF of the distance $|x - y|$ is given by
		\begin{align}
			f(r) = \frac{r}{\lambda_2(\mathcal{A})\lambda_2(\mathcal{B})}\int_0^{2\pi} \lambda_2(\mathcal{B}_{r,\theta} \cap \mathcal{A}) \mathrm{d}\theta, \label{eq:thm}
		\end{align}
		where $\mathcal{B}_{r,\theta} = \{x \in \mathbb{R}^2 : x - (r\cos(\theta), r\sin(\theta)) \in \mathcal{B}\}$ is the set $\mathcal{B}$ shifted by the vector $(r\cos(\theta), r\sin(\theta))$.
		\end{theorem}
	
		\begin{proof}
			Since $x$ and $y$ are uniformly distributed, their PDFs are respectively $f_{\mathcal{A}}(x) = \frac{1}{\lambda_2(\mathcal{A})}$ and $f_{\mathcal{B}}(y) = \frac{1}{\lambda_2(\mathcal{B})}$. For notational convenience, let $c = \frac{1}{\lambda_2(\mathcal{A})\lambda_2(\mathcal{B})}$. Substituting these into~\eqref{eq:fund} and applying Fubini's theorem we obtain
			\begin{align}
				F(r) = c\int_{\mathbb{R}^2} \left(\int_{\mathbb{R}^2} \mathbbm{1}_{C_r}(x, y) \mathrm{d}x\right)\mathrm{d}y.
			\end{align}
			Making the coordinate transformation $x \mapsto x + y$ yields
			\begin{align}
				F(r) = c\int_{\mathbb{R}^2} \left(\int_{\mathbb{R}^2} \mathbbm{1}_{C_r}(x + y, y) \mathrm{d}x\right)\mathrm{d}y.
			\end{align}
			Decomposing the inner integral into polar coordinates gives
			\begin{align}
				F(r) = c\int_{\mathbb{R}^2} \left(\int_0^\infty \int_0^{2\pi} s\mathbbm{1}_{C_r}(s\omega + y, y) \mathrm{d}\theta \mathrm{d}s\right)\mathrm{d}y,
			\end{align}
			where $\omega$ is the point on the unit circle at an angle of $\theta$ from the $x$-axis. From the definition of $\mathbbm{1}_{C_r}$, we can restrict our integral in $s$ to the range $[0, r]$, since the characteristic function will vanish outside of this range. Thus we have
			\begin{align}
				F(r) = c\int_{\mathbb{R}^2} \left(\int_0^r \int_0^{2\pi} s\mathbbm{1}_{C_r}(s\omega + y, y) \mathrm{d}\theta \mathrm{d}s\right)\mathrm{d}y.
			\end{align}
			The PDF for the distance distribution is defined as $f(r) = \frac{\mathrm{d}}{\mathrm{d}r}F(r)$, so differentiating under the integral and using the fundamental theorem of calculus we find that
			\begin{align}
				f(r) = c\int_{\mathbb{R}^2} \left(\int_0^{2\pi} r\mathbbm{1}_{C_r}(r\omega + y, y) \mathrm{d}\theta\right)\mathrm{d}y.
			\end{align}
			Interchanging the order of integration and factoring out the constant $r$ we can write this as
			\begin{align}
				f(r) = rc\int_0^{2\pi} \left(\int_{\mathbb{R}^2} \mathbbm{1}_{C_r}(r\omega + y, y) \mathrm{d}y\right)\mathrm{d}\theta.
			\end{align}
			The inner integral is the measure of the set $\{b \in \mathcal{B} : b + r\omega \in \mathcal{A}\}$, i.e., the set of points in $\mathcal{B}$ that when shifted a distance $r$ along the angle $\theta$ lie in $\mathcal{A}$. This is equivalent to the measure of the set $\mathcal{B}_{r,\theta} \cap \mathcal{A}$ where $\mathcal{B}_{r,\theta} = \{x \in \mathbb{R}^2 : x - (r\cos(\theta), r\sin(\theta)) \in \mathcal{B}\}$. Thus we have
			\begin{align}
				f(r) = \frac{r}{\lambda_2(\mathcal{A})\lambda_2(\mathcal{B})}\int_0^{2\pi} \lambda_2(\mathcal{B}_{r,\theta} \cap \mathcal{A}) \mathrm{d}\theta. \label{eq:polar}
			\end{align}
			The expression \eqref{eq:polar} is our desired result.
		\end{proof}

		\begin{remark} 	
			When using Theorem~\ref{thm:new} to perform computations and thereby obtain closed form results, it is simplest to first consider the case that $\mathcal{A}$ and $\mathcal{B}$ are both triangles. Once a procedure for triangles is established, we can extend to the case of arbitrary polygons by decomposing each polygon into triangles and performing a probabilistic sum\footnote{We adopt this idea of decomposing polygon regions into triangles from~\cite{ahmadi_2014}.}. This is summarised in the following proposition.
		\end{remark}

		\begin{prop} \label{prop:ext}
			Let $\mathcal{P}, \mathcal{Q} \subset \mathbb{R}^2$ be arbitrary polygons. Suppose that $\mathcal{P}$ can be decomposed into the $n$ disjoint triangles $\mathcal{T}_1^{(\mathcal{P})}, \ldots, \mathcal{T}_n^{(\mathcal{P})}$ and that $\mathcal{Q}$ can be decomposed into the $m$ disjoint triangles $\mathcal{T}_1^{(\mathcal{Q})}, \ldots, \mathcal{T}_m^{(\mathcal{Q})}$, i.e.,
			\begin{align}
				\mathcal{P} &= \bigcup_{i=1}^n \mathcal{T}_i^{(\mathcal{P})},\;\text{and}\\
				\mathcal{Q} &= \bigcup_{j=1}^m \mathcal{T}_j^{(\mathcal{Q})},
			\end{align}
			If we denote the PDF of the distance distribution of two uniformly randomly distributed points in the triangles $\mathcal{T}_1$ and $\mathcal{T}_2$ by $f_{\mathcal{T}_1, \mathcal{T}_2}(r)$, then the PDF of the distance distribution for two uniformly randomly distributed points in $\mathcal{P}$ and $\mathcal{Q}$ is given by
			\begin{align}
				f(r) &= \sum_{i=1}^n\sum_{j=1}^m \frac{f_{\mathcal{T}_i^{(\mathcal{P})}, \mathcal{T}_j^{(\mathcal{Q})}}(r)}{\lambda_2\left(\mathcal{T}^{(\mathcal{P})}_i \right) \lambda_2\left(\mathcal{T}^{(\mathcal{Q})}_j\right)}. \label{eq:sum}
			\end{align}
		\end{prop}
		\begin{proof}
			Using the decomposition of $\mathcal{P}$ and $\mathcal{Q}$ into triangles and \eqref{eq:meas_decomp} we can write
			\begin{align}
				\mathbb{P}(C_r) = \sum_{i=1}^n\sum_{j=1}^m \mathbb{P}\left(C_r \cap \left(\mathcal{T}^{(\mathcal{P})}_i \times \mathcal{T}^{(\mathcal{Q})}_j\right)\right). \label{eq:propproof}
			\end{align}
			Using conditional probability as defined by \eqref{eq:cond} we can write the RHS of \eqref{eq:propproof} as
			\begin{align}
				\sum_{i=1}^n\sum_{j=1}^m \mathbb{P}\left(\mathcal{T}^{(\mathcal{P})}_i \times \mathcal{T}^{(\mathcal{Q})}_j\right)\mathbb{P}\left(C_r \left| \mathcal{T}^{(\mathcal{P})}_i \times \mathcal{T}^{(\mathcal{Q})}_j\right.\right).
			\end{align}
			Since $\mathbb{P}(C_r)$ is the CDF of the distance distribution for $\mathcal{P}$ and $\mathcal{Q}$, we find the PDF to be the derivative, namely
			\begin{align}
				\sum_{i=1}^n\sum_{j=1}^m \mathbb{P}\left(\mathcal{T}^{(\mathcal{P})}_i \times \mathcal{T}^{(\mathcal{Q})}_j\right)\frac{\mathrm{d}}{\mathrm{d}r}\mathbb{P}\left(C_r \left| \mathcal{T}^{(\mathcal{P})}_i \times \mathcal{T}^{(\mathcal{Q})}_j\right.\right).
			\end{align}
			But $\frac{\mathrm{d}}{\mathrm{d}r}\mathbb{P}\left(C_r \left| \mathcal{T}^{(\mathcal{P})}_i \times \mathcal{T}^{(\mathcal{Q})}_j\right.\right)$ is precisely the PDF $f_{\mathcal{T}_i^{(\mathcal{P})}, \mathcal{T}_j^{(\mathcal{Q})}}(r)$, and substituting also
			\begin{align}
				\mathbb{P}\left(\mathcal{T}^{(\mathcal{P})}_i \times \mathcal{T}^{(\mathcal{Q})}_j\right) = \frac{1}{\lambda_2\left(\mathcal{T}^{(\mathcal{P})}_i \right) \lambda_2\left(\mathcal{T}^{(\mathcal{Q})}_j\right)}
			\end{align}
			we obtain \eqref{eq:sum} as required.
		\end{proof}

\section{Results}\label{sec:results}
With Theorem~\ref{thm:new} established we will now confirm its validity by comparing its predictions to established results and to simulation, and we will also use it to obtain new results that cover cases for which results in the literature do not exist.
	
\subsection{Two Coincident Circles}
First, consider the simple case that $\mathcal{A} = \mathcal{B}$ are circles of equal radius $R$. The PDF in this case has been computed in closed form by Mathai~\cite{mathai_1999}.

To use Theorem~\ref{thm:new}, we use the fact that the area of intersection between two circles of equal radius $R$ and centres separated by a distance $r$ (assuming they intersect) is

		\begin{align}
			2R^2\cos^{-1}\left(\frac{r}{2R}\right) - \frac{r}{2}\sqrt{4R^2 - r^2}.
		\end{align}

But for any $\theta \in [0, 2\pi)$ this is exactly $\lambda_2(\mathcal{B}_{r,\theta} \cap \mathcal{A})$ and so substituting into \eqref{eq:thm} gives the PDF as
		\begin{align}
			\frac{r}{\pi^2 R^4}\int_0^{2\pi} 2R^2\cos^{-1}\left(\frac{r}{2R}\right) - \frac{r}{2}\sqrt{4R^2 - r^2}\mathrm{d}\theta.
		\end{align}

There will be intersection for $r < 2R$ so we can compute this as
		\begin{align}
			\begin{cases}
				\frac{2r}{\pi R^4}\left(2R^2\cos^{-1}\left(\frac{r}{2R}\right) - \frac{r}{2}\sqrt{4R^2 - r^2}\right) & r < 2R\\
				0 & \text{otherwise}
			\end{cases}
		\end{align}

This agrees with the result given by Mathai~\cite{mathai_1999}.
		
	\subsection{Two Disjoint Triangles}
		
		\begin{figure}[t]
			\centering
			\resizebox{0.4\textwidth}{!}{%
				\includegraphics{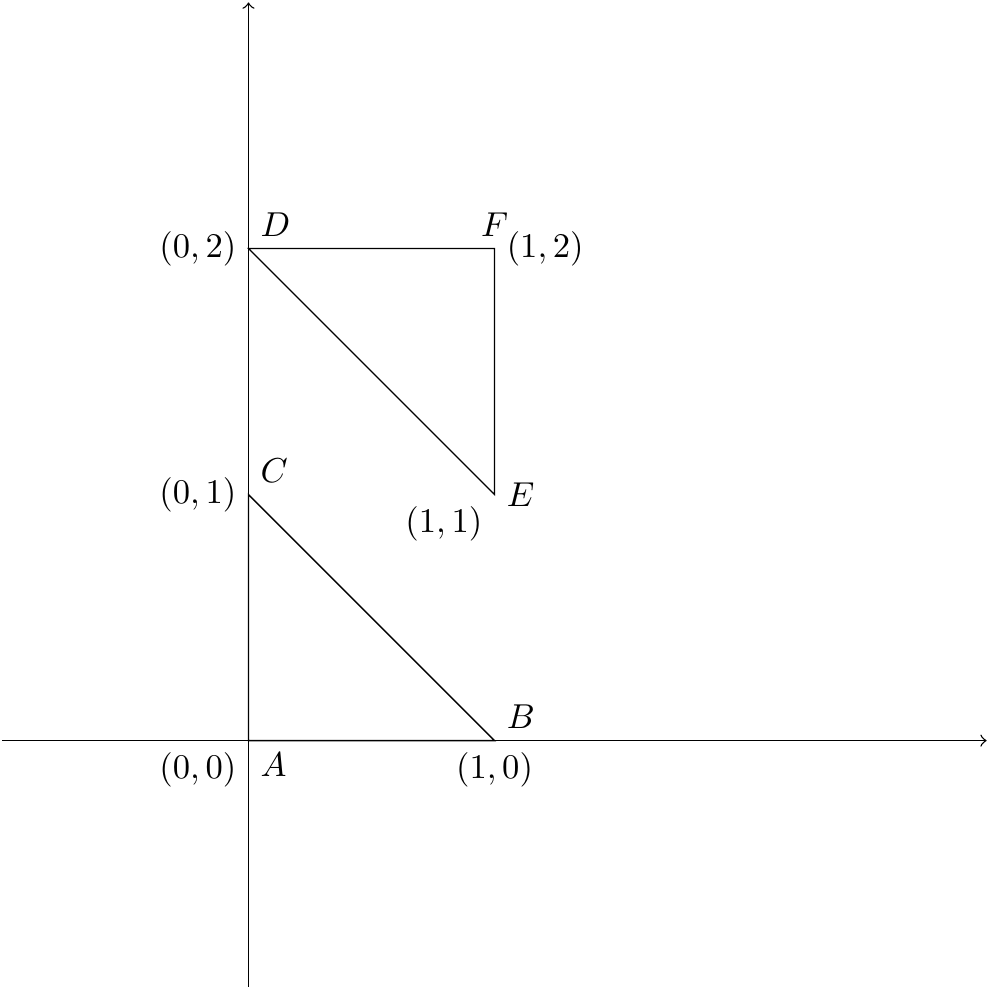}
			}%
				\caption{{Example Geometry for Closed Form Computation.} Example case for computation of distance distributions using Theorem~\ref{thm:new}. Here, $A = \triangle ABC$ and $B = \triangle DEF$.}
			\label{fig:triangles}
		\end{figure}

		Next we can obtain a closed form result for a case that, to the best of our knowledge, does not yet exist in the literature; two disjoint triangles. For example, consider the case that $\mathcal{A} = \triangle ABC$ and $\mathcal{B} = \triangle DEF$ shown in Fig.~\ref{fig:triangles}. In this case, the PDF is found to be (the computation is outlined in Appendix~\ref{sec:comp})
		\begin{align}
			f(r) &=
			\begin{cases}
				f_1(r) & r \in \left[\frac{1}{\sqrt{2}}, 1\right]\\
				f_2(r) & r \in [1, \sqrt{2}]\\
				f_3(r) & r \in [\sqrt{2}, 2]\\
				f_4(r) & r \in [2, \sqrt{5}]\\
				0 & \text{otherwise}
			\end{cases}, \label{eq:ex}
		\end{align}
		where $f_1$, $f_2$, $f_3$ and $f_4$ are given by

\begin{small}
			\begin{align}
				f_1(r) &= \frac{1}{8} \left(4 \sqrt{2 r^2-1}+r \sqrt{4 r^2-2} \sqrt{\frac{\sqrt{2 r^2-1}}{r^2}+1}+7 r \sqrt{\frac{2 \sqrt{2 r^2-1}}{r^2}+2} + r \sqrt{4 r^2-2} \sqrt{1-\frac{\sqrt{2 r^2-1}}{r^2}}\right. \nonumber \\%
				&\left. -7 r \sqrt{2-\frac{2 \sqrt{2 r^2-1}}{r^2}}+4 \left(r^2+2\right) \cos ^{-1}\left(\frac{\sqrt{2 r^2-1}+1}{2 r}\right)\right) - \frac{1}{2} \left(r^2+2\right) \cos ^{-1}\left(-\frac{\sqrt{2 r^2-1}-1}{2 r}\right)\\
				f_2(r) &= \frac{r^2}{2}+\frac{\pi  r^2}{4}-3 \sqrt{1-\frac{1}{r^2}} r+r^2 \left(-\csc ^{-1}(r)\right)+\frac{1}{8} \left(2 \left(5 \sqrt{2 r^2-1}-6\right)+ r \left(-\sqrt{4 r^2-2} \sqrt{\frac{\sqrt{2 r^2-1}}{r^2}+1} \right.\right. \nonumber \\
				&\left.\left. +9 \sqrt{\frac{2 \sqrt{2 r^2-1}}{r^2}+2}+2 r-16\right)+ 4 \left(r^2+4\right) \csc ^{-1}\left(\frac{2 r}{1-\sqrt{2 r^2-1}}\right)\right) \nonumber\\
				&+\frac{1}{2} \left(\left(6 \sqrt{1-\frac{1}{r^2}}-r+4\right) r-2 \left(r^2+2\right) \sec ^{-1}(r)-3\right)- \sin ^{-1}\left(\frac{1}{r}\right)+\cos ^{-1}\left(\frac{1}{r}\right)+2 \\
				f_3(r) &= 2 r-2 \cos ^{-1}\left(\frac{1}{r}\right)+\frac{7}{2}-\frac{r^2}{2}-2 \sqrt{\frac{4 \sqrt{2} \sqrt{r^2-2}}{r^2}+2} r+2 \sqrt{1-\frac{1}{r^2}} r-\frac{1}{2} \pi  \left(r^2+2\right)-\sqrt{2} \sqrt{r^2-2} \nonumber\\
				&+\left(r^2+2\right) \cos ^{-1}\left(-\frac{\sqrt{2 r^2-4}-2}{2 r}\right)+2 \cos ^{-1}\left(-\frac{\sqrt{2 r^2-4}-2}{2 r}\right)+\frac{1}{8} \left(2 \left(5 \sqrt{2 r^2-1}-6\right)+\right. \nonumber\\
				&\left. r \left(-\sqrt{4 r^2-2} \sqrt{\frac{\sqrt{2 r^2-1}}{r^2}+1}+9 \sqrt{\frac{2 \sqrt{2 r^2-1}}{r^2}+2}+2 r-16\right)+4 \left(r^2+4\right) \csc ^{-1}\left(\frac{2 r}{1-\sqrt{2 r^2-1}}\right)\right) \\
				f_4(r) &= -\frac{r^2}{2}+2 \sqrt{1-\frac{1}{r^2}} r+\sqrt{1-\frac{4}{r^2}} r+\frac{1}{8} \left(-2 r^2-\sqrt{4 r^2-2} \sqrt{\frac{\sqrt{2 r^2-1}}{r^2}+1} r+9 \sqrt{\frac{2 \sqrt{2 r^2-1}}{r^2}+2} r \right. \nonumber\\
				&\left. -16 \sqrt{1-\frac{4}{r^2}} r+10 \sqrt{2 r^2-1}+4 \left(r^2+4\right) \sin ^{-1}\left(\frac{1-\sqrt{2 r^2-1}}{2 r}\right)+ 4 \left(r^2+4\right) \cos ^{-1}\left(\frac{2}{r}\right)-28\right) \nonumber\\
				&+2 \sin ^{-1}\left(\frac{2}{r}\right)-2 \cos ^{-1}\left(\frac{1}{r}\right)-\frac{5}{2}.
			\end{align}
		\end{small}

		To validate \eqref{eq:ex}, we compare the computed theoretical PDF with one that was simulated and also to the numerical result obtained using the kinematic measure technique in~\cite{tong_2017}. To obtain a simulated PDF, 4000 points were chosen uniformly randomly inside each triangle (these triangles are depicted in Fig.~\ref{fig:triangles}), resulting in a total of 8000 simulated points, and each of the distances between pairs of points from each triangle were computed. This means that a total of $4000^2$ random distances were obtained, from which the simulated PDF was estimated using a kernel density estimation implemented in Mathematica. For the technique in~\cite{tong_2017}, their Matlab implementation was used. The graph of the computed function \eqref{eq:ex}, points from the kernel density estimation and the numerical result using the technique in~\cite{tong_2017} are shown in Fig.~\ref{fig:sim}. The results match, validating \eqref{eq:ex}.
		
Note that using similar techniques as outlined in Appendix~\ref{sec:comp}, it is possible to compute in closed form the PDF for any two arbitrary triangles. This is important because in conjunction with Proposition~\ref{prop:ext}, it enables us to compute the closed form PDF in the case that $\mathcal{A}$ and $\mathcal{B}$ are arbitrary polygons by decomposing each into triangles and performing the probabilistic sum.

		\subsection{Two Arbitrary Polygons}
		Consider finally a more complicated example depicted in Fig.~\ref{fig:polar}, with the corresponding PDFs plotted in Fig.~\ref{fig:polar_plot}. In this case, for simplicity, the integration in \eqref{eq:thm} was performed numerically using Mathematica. The mathematica code is provided in Appendix~\ref{sec:code} and is also downloadable from~\cite{durrani_2018}. The Mathematica code allows a user to define any two arbitrary polygons and determine the distance distribution numerically.
		
		Obtaining numerical results highlights one point of distinction between our proposed technique and the technique presented in~\cite{tong_2017}. Both can be used to obtain numerical results for arbitrary polygons. However, for the technique in~\cite{tong_2017} it is required to first decompose the polygon into triangles and then perform a probabilistic sum, such as outlined in Proposition~\ref{prop:ext}. In contrast, our proposed technique can use the fundamental result in~\eqref{eq:thm} directly for any arbitrary polygons.

		\begin{figure}[t]
			\centering
			\resizebox{0.4\textwidth}{!}{%
				\includegraphics{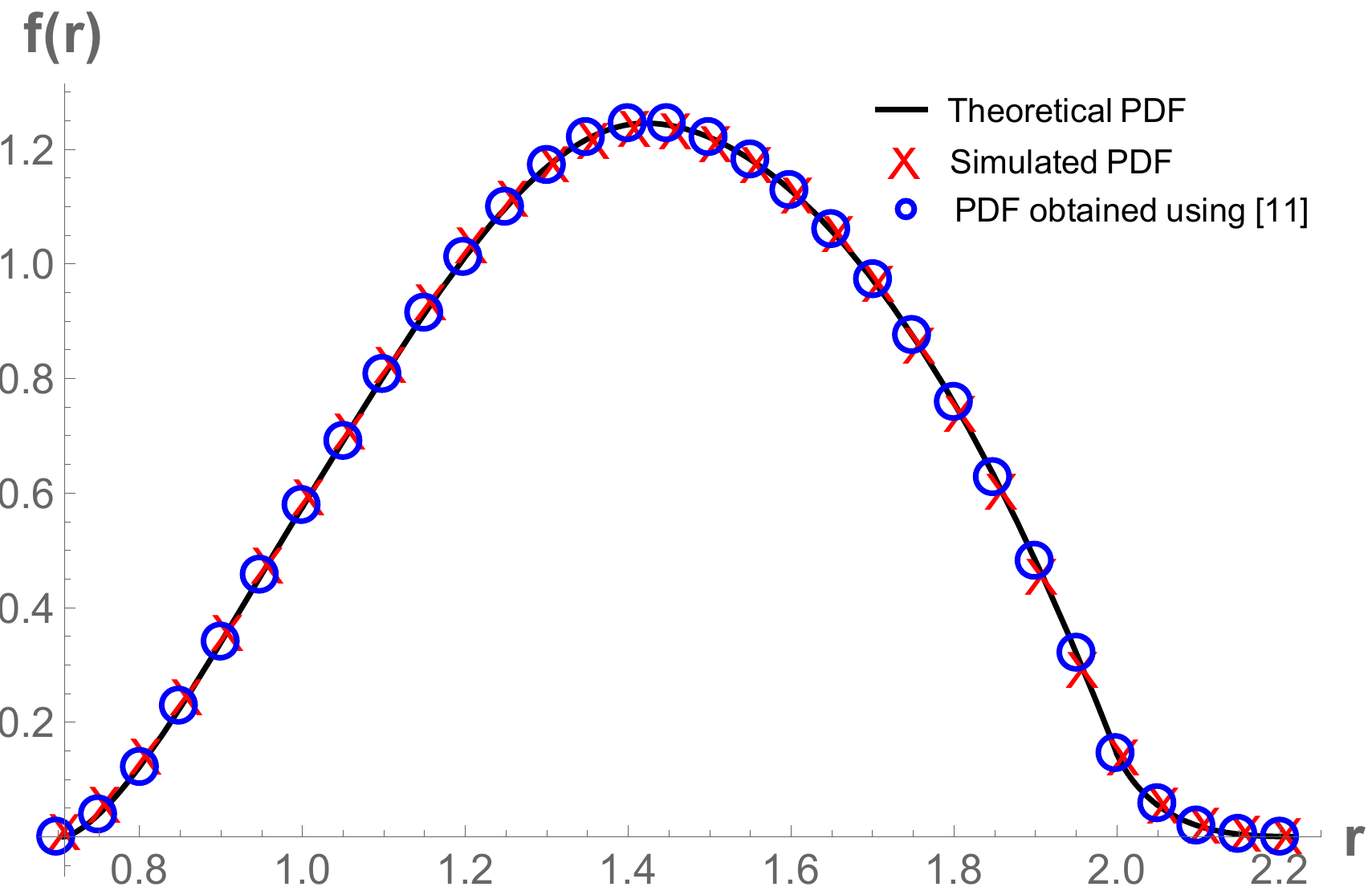}
			}%
			\caption{{Theoretical versus Simulated PDF Comparison.} Plot showing the graph of the theoretical PDF (black line), the simulated PDF (red crosses) and the numerical result from Tong and Pan's Matlab implementation (blue circles).}
			\label{fig:sim}
		\end{figure}
		
\section{Conclusion}\label{sec:conc}
	In this paper, we have proposed a novel framework for obtaining distance distribution between two random nodes in arbitrary polygons. The proposed framework uses measure theory and allows us to obtain closed form results for cases that have not yet been reported in the literature. We have also developed a Mathematica implementation of the proposed framework. Future work can extend this Mathematica implementation to automatically compute the distance distribution between two random nodes in arbitrary polygons in closed form, similar to~\cite{pure_2015} which uses Mathematica to compute the closed form distance distribution between a fixed reference point and one random node in an arbitrary polygon.

		\begin{figure}[t]
			\centering
			\resizebox{0.4\textwidth}{!}{%
				\includegraphics{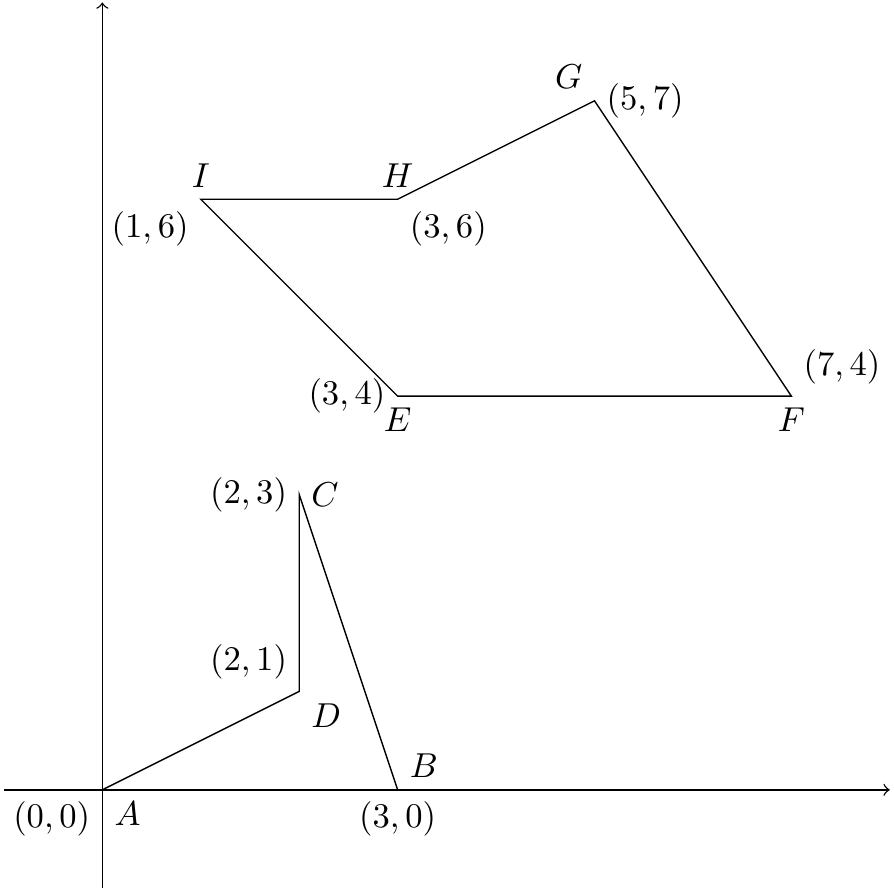}
			}%
				\caption{{Example Geometry for Numerical Computation.} Example pair of arbitrary polygons used to validate the integral formula \eqref{eq:fund}.}
			\label{fig:polar}
		\end{figure}
		
		\begin{figure}[t]
			\centering
			\resizebox{0.4\textwidth}{!}{%
				\includegraphics{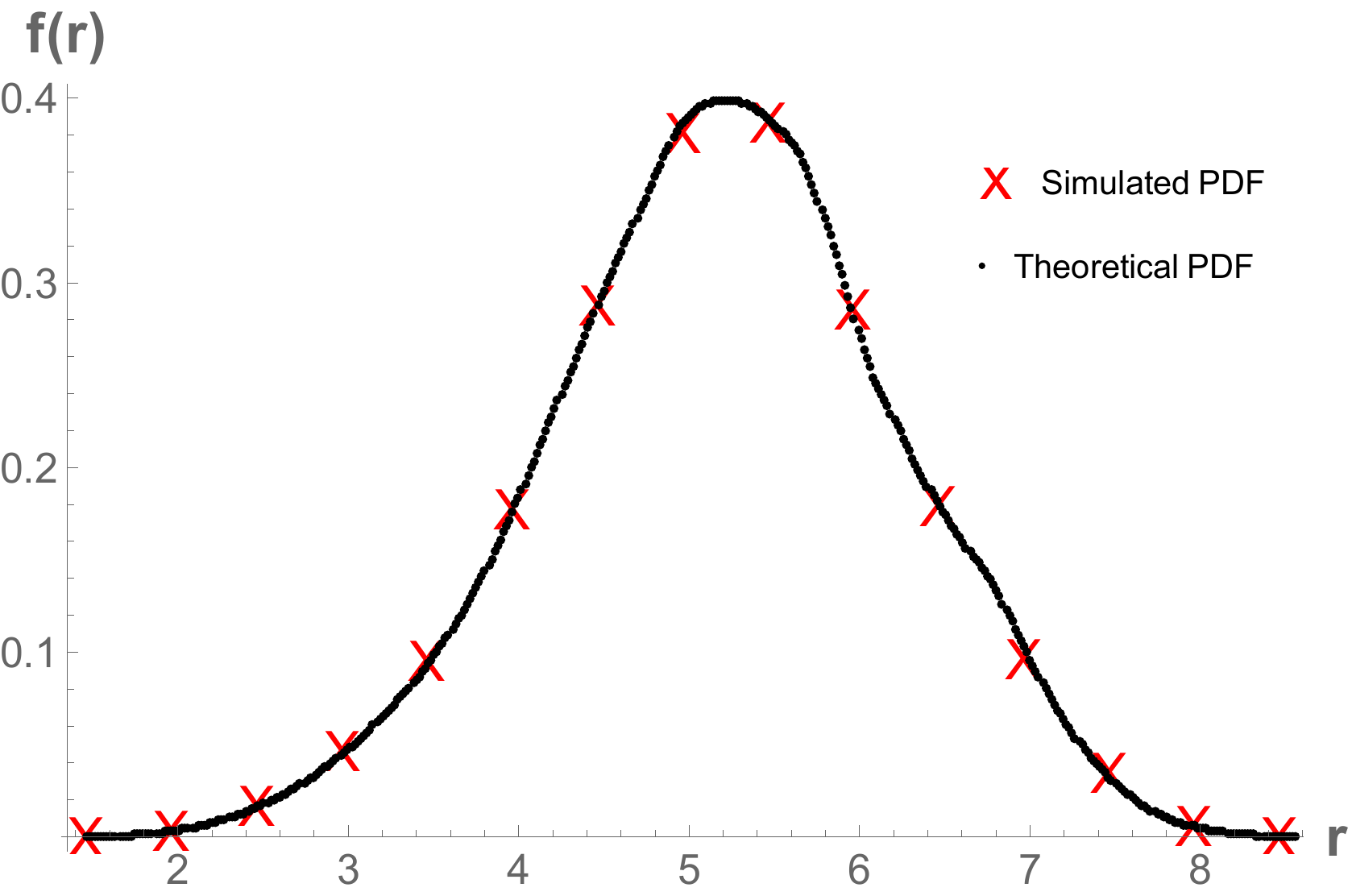}
			}%
			\caption{{Theoretical versus Simulated PDF Comparison.} Theoretical (black) and simulated PDF (red crosses) for the polygons shown in Figure~\ref{fig:polar}.}
			\label{fig:polar_plot}
		\end{figure}






\section{Appendices}

\subsection{Example Computation of Distance Distribution} \label{sec:comp}

	To use \eqref{eq:fund} we need to determine $\lambda_2(\mathcal{B}_{r,\theta} \cap \mathcal{A})$ explicitly. Clearly $\mathcal{B}_{r,\theta} \cap \mathcal{A}$ for any fixed $r$ and $\theta$ will be a polygon, and so to compute the area we need only determine the vertices at which point we can use the Shoelace Formula stated in Theorem~\ref{thm:shoelace}.
	
	To compute $\lambda_2(\mathcal{B}_{r,\theta} \cap \mathcal{A})$ we need to find what the vertices of the polygon $\mathcal{B}_{r,\theta} \cap \mathcal{A}$ are, which will depend on both $r$ and $\theta$. There will be different ranges of $r$ and $\theta$ for which the vertices will have the same expression, and so these ranges need to be determined. After this, the explicit form of $\lambda_2(\mathcal{B}_{r,\theta} \cap \mathcal{A})$ can be found using the Shoelace formula, which can then be integrated to determine the PDF using \eqref{eq:fund}.
	
	The vertices of $\lambda_2(\mathcal{B}_{r,\theta} \cap \mathcal{A})$ will either be vertices of $\mathcal{A}$, vertices of $\mathcal{B}_{r,\theta}$, or the points of intersection of sides from $\mathcal{A}$ and $\mathcal{B}_{r,\theta}$. Vertices of $\mathcal{A}$ are constant and are simply $(0,2)$, $(1,1)$, and $(1,2)$. The vertices of $\mathcal{B}_{r,\theta}$ are shifted copies of vertices from $\mathcal{B}$, and so are simply
	\begin{align}
		&(r\cos(\theta), r\sin(\theta)),\\
		&(1 + r\cos(\theta), r\sin(\theta)),\;\text{and}\\
		&(r\cos(\theta), 1 + r\sin(\theta)).
	\end{align}
	We will denote these vertices $A'$, $B'$ and $C'$ respectively. To find the intersection of two lines, we can utilise an explicit formula. If the first line is determined by the points $(x_1, y_1)$ and $(x_2, y_2)$, and the second line is determined by the points $(x_3, y_3)$ and $(x_4, y_4)$, then their point of intersection (assuming the lines are not parallel) is $(x, y)$, where
	\begin{align}
		x &= \frac{(x_1y_2 - x_2y_1)(x_3 - x_4) - (x_1 - x_2)(x_3y_4 - x_4y_3)}{(x_1 - x_2)(y_3 - y_4) - (y_1 - y_2)(x_3 - x_4)},\;\text{and}\\
		y &= \frac{(x_1y_2 - x_2y_1)(y_3 - y_4) - (y_1 - y_2)(x_3y_4 - x_4y_3)}{(x_1 - x_2)(y_3 - y_4) - (y_1 - y_2)(x_3 - x_4)}.
	\end{align}
	Using this, we can find all of the points of intersection of the sides from $\mathcal{A}$ and $\mathcal{B}_{r,\theta}$. They are (ignoring pairs of sides that are parallel)
	\begin{align}
		AB\;\text{and}\;DE&: \quad (2 - r\sin(\theta), r\sin(\theta)),\\
		AB\;\text{and}\;EF&: \quad (1, r\sin(\theta)),\\
		BC\;\text{and}\;EF&: \quad (1, r(\cos(\theta) + \sin(\theta))),\\
		BC\;\text{and}\;FD&: \quad (r(\cos(\theta) + \sin(\theta)) - 1, 2),\\
		CA\;\text{and}\;DE&: \quad (r\cos(\theta), 2 - r\cos(\theta)),\\
		CA\;\text{and}\;FD&: \quad (r\cos(\theta), 2).
	\end{align}
	We will denote these vertices $V_{11}$, $V_{12}$, $V_{22}$, $V_{23}$, $V_{31}$ and $V_{33}$ respectively. Next we determine the ranges of $r$ and $\theta$ for which each of the above 12 vertices constitute $\mathcal{B}_{r,\theta} \cap \mathcal{A}$. Since the vertices and hence also the end points of the line segments for $\mathcal{B}$ are translated around a circle of radius $r$, it is useful to have a formula for the points of intersection of a circle and a line. If the line is determined by the points $(x_1, y_1)$ and $(x_2, y_2)$, and if we define
	\begin{align}
		d_x &= x_2 - x_1,\\
		d_y &= y_2 - y_1,\\
		d_r &= \sqrt{d_x^2 + d_y^2},\;\text{and}\\
		D &= \begin{vmatrix}x_1 & x_2 \\ y_1 & y_2\end{vmatrix},
	\end{align}
	then the points of intersection are given by
	\begin{align}
		x &= \frac{D d_y \pm \text{sgn}^\ast(d_y)d_x\sqrt{r^2 d_r^2 - D^2}}{d_r^2}\\
		y &= \frac{-D d_x \pm |d_y|\sqrt{r^2 d_r^2 - D^2}}{d_r^2},
	\end{align}
	where
	\begin{align}
		\text{sgn}^\ast(x) =
		\begin{cases}
			-1 & x < 0\\
			1 & x \ge 0
		\end{cases}.
	\end{align}
	
	Using this, the correspondences are found to be as follows.
	\begin{itemize}
		\item $r \in \left[\frac{1}{\sqrt{2}}, 1\right]$:
		\begin{align}
			&\{V_{22}, C', V_{31}, E\}: \\
			&\theta \in \left[\cos^{-1}\left(\frac{1 + \sqrt{2r^2 - 1}}{2r}\right), \cos^{-1}\left(\frac{1 - \sqrt{2r^2 - 1}}{2r}\right)\right]
		\end{align}
		\item $r \in [1, \sqrt{2}]$:
		\begin{align}
			&\{V_{22}, C', V_{31}, E\}: \;  \theta \in \left[\cos^{-1}\left(\frac{1}{r}\right), \sin^{-1}\left(\frac{1}{r}\right)\right]\\
			&\{V_{11}, V_{12}, V_{22}, V_{23}, V_{33}, V_{31}\}: \;  \theta \in \left[\sin^{-1}\left(\frac{1}{r}\right), \frac{\pi}{2}\right]\\
			&\{B', V_{23}, D, V_{11}\}: \;  \theta \in \left[\frac{\pi}{2}, \frac{\pi}{2} + \sin^{-1}\left(-\frac{1 - \sqrt{2r^2 - 1}}{2r}\right)\right]
		\end{align}
		\item $r \in [\sqrt{2}, 2]$:
		\begin{align}
			&\{V_{12}, F, V_{33}, A'\}:\\
			&\theta \in \left[\cos^{-1}\left(\frac{1}{r}\right), \cos ^{-1}\left(\frac{2-\sqrt{2 r^2-4}}{2 r}\right)\right]\\
			&\{V_{11}, V_{12}, V_{22}, V_{23}, V_{33}, V_{31}\}:\\
			&\theta \in \left[\cos ^{-1}\left(\frac{2-\sqrt{2 r^2-4}}{2 r}\right), \frac{\pi}{2}\right]\\
			&\{B', V_{23}, D, V_{11}\}:\\
			&\theta \in \left[\frac{\pi}{2}, \frac{\pi}{2} + \sin^{-1}\left(-\frac{1 - \sqrt{2r^2 - 1}}{2r}\right)\right]
		\end{align}
		\item $r \in [2, \sqrt{5}]$:
		\begin{align}
				&\{V_{12}, F, V_{33}, A'\}: \; \theta \in \left[\cos^{-1}\left(\frac{1}{r}\right), \sin ^{-1}\left(\frac{2}{r}\right)\right]\\
				&\{B', V_{23}, D, V_{11}\}: \\
				&\theta \in \left[\frac{\pi}{2} + \cos ^{-1}\left(\frac{2}{r}\right), \frac{\pi}{2} + \sin ^{-1}\left(-\frac{1-\sqrt{2 r^2-1}}{2 r}\right)\right]
		\end{align}
	\end{itemize}
	This is all the required information to determine $\lambda_2(\mathcal{B}_{r,\theta} \cap \mathcal{A})$ using the Shoelace Formula. Substituting this into the integral in \eqref{eq:fund} yields the result \eqref{eq:ex}.

\subsection{Mathematica Code to Compute Distance Distribution} \label{sec:code}

	\begin{lstlisting}[caption={Simulation Functions}]
		(*PolygonArea computes the area of a polygon using the Shoelace Formula*)
		PolygonArea[P_] :=
		 1/2 Total[Det /@ Partition[Append[P, First@P], 2, 1]]
		(*RandomPointsTriangle generates uniformly random points in a triangle*)
		RandomPointsTriangle[{a_, b_, c_}, n_] := Module[{u, v},
		  {u, v} = Transpose[Sort /@ RandomReal[1, {n, 2}]];
		  Map[# a &, u] + Map[# b &, (v - u)] + Map[# c &, (1 - v)]
		  ]
		(*RandomPointsPolygon generates uniformly random points in a polygon*)
		
		RandomPointsPolygon[P_, n_] := Module[{p, s, t},
		  p = N[P];(*Numerically approximate vertices for faster computation.*)
		  s =
		   First /@
		    MeshPrimitives[
		     TriangulateMesh[
		      DiscretizeGraphics@Graphics[Polygon@p],
		      MaxCellMeasure -> Infinity],
		     2];
		  t = Accumulate[PolygonArea[#]/PolygonArea[p] & /@ s];
		  (*Triangulate the polygon. Calculate the area of the polygon.	Associate each triangle with its fraction of area of the polygon. Associate each triangle with a range calculated as the sum of all previous area fractions. This allows a triangle to be picked from a random variable generated between 0 and 1.*)
		  Select[
		   Map[
		    Function[
		     x,
		     Flatten@RandomPointsTriangle[
		       s[[First@FirstPosition[x < # & /@ t, True]]], 1]
		     ],
		    RandomReal[1, n]
		    ],
		   # != {} &]
		  (*Random points: Generate n random numbers between 0 and 1. For each random number, pick the corresponding triangle and generate a point in it.*)
		  ]
	\end{lstlisting}
	
	\begin{lstlisting}[caption={Random Data Generation}]
		n = 4000;
		Poly1 = {{0, 0}, {3, 0}, {2, 3}, {2, 1}};
		Poly2 = {{3, 4}, {7, 4}, {5, 7}, {3, 6}, {1, 6}};
		data1 = RandomPointsPolygon[Poly1, n];
		data2 = RandomPointsPolygon[Poly2, n];
	\end{lstlisting}
	
	\begin{lstlisting}[caption={Computing the Simulated PDF}]
		data = Norm[#[[1]] - #[[2]]] & /@ Tuples[{data1, data2}];
		dist = SmoothKernelDistribution@data;
		min = Min@data;
		max = Max@data;
	\end{lstlisting}
	
	\begin{lstlisting}[caption={Area of Intersection Function}]
		areatest[r_, \[Theta]_] := Area[RegionIntersection @@ Polygon /@ {N[# + {r Cos[\[Theta]], r Sin[\[Theta]]}] & /@ Poly1, N /@ Poly2}]
	\end{lstlisting}
	
	\begin{lstlisting}[caption={Perform Numerical Integration and Plot Result}]
		pts = 10;(*Number of PDF points to simulate; 10 points takes about a minute*)
		divs = 100;
		t = Table[{r, (2 \[Pi] r)/(divs (Area@Polygon@Poly1) (Area@Polygon@Poly2))
		      Total@Table[
		       areatest[r, \[Theta]], {\[Theta], 0, 2 \[Pi], (2 \[Pi])/divs}]}, {r, min, max, (max - min)/pts}];
		Show[
		 Plot[PDF[dist, x], {x, min, max}, PlotStyle -> {Black, Thick},
		  PlotLegends -> Placed[{"Numerical PDF"}, {0.85, 0.82}]],
		 ListPlot[t, PlotMarkers -> Style["X", {Red, FontSize -> 20}],
		  PlotLegends -> Placed[{"Simulated PDF"}, {0.85, 0.82}]],
		 AxesLabel -> {"r", "f(r)"}, LabelStyle -> Directive[Bold, 20],
		 TicksStyle -> Directive[Plain, 15]
		 ]
	\end{lstlisting}

\end{document}